\pdfoutput=1
\newif\ifFull
\Fullfalse
\ifFull
\documentclass[11pt]{article}
\usepackage[margin=1.25in]{geometry}
\else
\documentclass{llncs}
\renewenvironment{proof}{\noindent{\bf Proof:}}{\hspace*{\fill}\qed\bigskip}
\fi

\usepackage{graphicx,wrapfig,subfig}
\usepackage{amsfonts,amsmath,amssymb}

\ifFull
\usepackage{mathptm}
\usepackage{amsthm}
\newtheorem{theorem}{Theorem}[section]
\newtheorem{lemma}[theorem]{Lemma}

\fi
\usepackage{hyperref}
\usepackage{cite}

\title{Planar Lombardi Drawings for Subcubic Graphs}
\ifFull
\author{David Eppstein \\ 
 \\
 Department of Computer Science, University of California, Irvine, USA.}
\else
\author{David Eppstein}
\institute{Department of Computer Science, University of California, 
	   Irvine, USA.}
\fi

\begin{document}

\maketitle

\pagestyle{plain}

\begin{abstract}
We prove that every planar graph with maximum degree three has a planar drawing in which the edges are drawn as circular arcs that meet at equal angles around every vertex.  Our construction is based on the Koebe--Thurston--Andreev circle packing theorem, and uses a novel type of Voronoi diagram for circle packings that is invariant under M\"obius transformations, defined using three-dimensional hyperbolic geometry. We also use circle packing to construct planar Lombardi drawings of a special class of 4-regular planar graphs, the medial graphs of polyhedral graphs, and we show that not every 4-regular planar graph has a Lombardi drawing. We have implemented our algorithm for 3-connected planar cubic graphs.
\end{abstract}

\section{Introduction}

\emph{Lombardi drawing} is a style of graph drawing, named after artist Mark Lombardi, in which the edges of a graph are drawn as circular arcs and in which every vertex is surrounded by edges that meet at equal angles at the vertex---that is, the drawing has \emph{perfect angular resolution}~\cite{DunEppGoo-GD-10a,DunEppGoo-GD-10b}. Several families of graphs are known to have Lombardi drawings, including regular graphs (under certain restrictions on their factorizations into regular subgraphs) and certain highly symmetric graphs~\cite{DunEppGoo-GD-10b}. Lombardi drawings have also been used to draw plane trees with perfect angular resolution in polynomial area, something that would be impossible for straight line drawings~\cite{DunEppGoo-GD-10a}.
All graphs have drawings that relax the constraints of Lombardi drawing to allow polyarcs or unequal angles~\cite{CheCunGoo-GD-11,DunEppGoo-GD-11}, but the improved aesthetic quality of true Lombardi drawings makes it of interest to determine more precisely which graphs have such drawings.

\emph{Planarity}, the avoidance of crossing edges, is of great importance both in Lombardi drawing and in graph drawing more generally. By F\'ary's theorem, every planar graph can be drawn planarly with straight line segments for its edges, and therefore it can also be drawn with circular arcs. However, those arcs may not meet at equal angles. Indeed, our past work has found planar graphs all of whose Lombardi drawings have crossings~\cite{DunEppGoo-GD-10b,DunEppGoo-GD-11}. Very few positive results on planar Lombardi drawing are known: only trees, Halin graphs (the graphs formed from plane trees by adding a cycle connecting the leaves) and certain highly symmetric planar graphs have been proven to have planar Lombardi drawings~\cite{DunEppGoo-GD-10b}.

In this paper we take a major step forward in our knowledge of planar Lombardi drawings, and in the applicability of the  Lombardi drawing style, by showing that all planar graphs of maximum degree three have planar Lombardi drawings.
The heart of our method applies to 3-connected 3-regular planar graphs: we apply the Koebe--Andreev--Thurston circle packing theorem to the dual graph, and then use three-dimensional hyperbolic geometry to construct a novel M\"obius-invariant Voronoi diagram of the circle packing, which we show to be a Lombardi drawing of the input. Our implementation of this method produces drawings in which the overall drawing and the individual faces are all approximately circular and in which the spacing of the vertices is locally uniform. We extend these results to graphs that are neither 3-regular nor 3-connected by using  bridge-block trees and SPQR trees to decompose the graph into 3-connected subgraphs, drawing these subgraphs separately, and using M\"obius transformations to glue them together into a single drawing. We also use circle packing in a different way to construct Lombardi drawings of a special class of 4-regular planar graphs, the medial graphs of polyhedral graphs. However, as we show, not every 4-regular planar graph has a Lombardi drawing.

\section{Preliminaries}

\subsection{M\"obius transformations}

Let ${\mathbb{S}}^2$ denote the space formed by adding a single point $\infty$ ``at infinity'' to the Euclidean plane; this space is also known as the one-dimensional complex projective line $\mathbb{P}^1(\mathbb{C})$. In ${\mathbb{S}}^2$, straight lines may be interpreted as limiting cases of circles, with infinite radius and containing the point $\infty$.
A \emph{M\"obius transformation}~\cite{Sch-79} is a map from ${\mathbb{S}}^2$  to itself that transforms every circle (or line) into another circle (or line). Using complex-number coordinates, these transformations may be represented as the \emph{fractional linear transformations}
\ifFull
\else
and their conjugates.
\fi
\ifFull
\else
Every
\fi
two triangles may be mapped to each other by M\"obius transformations.

M\"obius transformations may also be defined from \emph{inversions}. If $O$ is a circle centered at  point~$o$ with radius $r$, inversion through $O$ maps any point $p$ to another point $q$ on the ray from $o$ through $p$, at distance $r^2/d(o,p)$ from~$o$. $O$~is fixed by the inversion, the inside of $O$ becomes the outside and vice versa, and $o$ trades places with $\infty$. In the limiting case of a line, inversion becomes reflection across the line. Every M\"obius transformation may be represented as a composition of a finite set of inversions.

M\"obius transformations are \emph{conformal mappings}: they preserves the angle of every two incident curves. Since they preserve both circularity and angles, they preserve the property of being a Lombardi drawing.

\subsection{Hyperbolic geometry}

To understand our drawing algorithm, it will be helpful to understand some qualitative features of hyperbolic geometry, avoiding detailed calculations~\cite{CanFloKen-FoG-97}.

In the \emph{upper halfspace model} of three-dimensional hyperbolic space,  hyperbolic space is represented by an open halfspace of three-dimensional Euclidean space, but with a non-Euclidean distance metric.  The boundary plane of the Euclidean halfspace does not belong to the hyperbolic space but may be thought of as the set of ``points at infinity'' for the hyperbolic space. It is convenient to add one more point $\infty$ to this boundary plane, so that it becomes a copy of ${\mathbb{S}}^2$.
Hyperbolic lines are represented by Euclidean semicircles that meet the boundary plane at right angles, or by Euclidean rays perpendicular to the boundary plane; the vertical rays may be viewed as the limiting case of semicircles with one endpoint at~$\infty$.
Hyperbolic planes are represented by Euclidean hemispheres that meet the boundary plane at right angles in a circle, or by halfplanes that touch the point $\infty$ and meet the boundary plane perpendicularly in a line.

Hyperbolic space is locally Euclidean: within sufficiently small balls, hyperbolic distances may be approximated arbitrarily well by Euclidean distances. For every two reference frames (a choice of a point within the space, and a system of Cartesian coordinates for the local Euclidean geometry near that point) a unique symmetry of the space takes one reference frame into the other. Symmetries of hyperbolic space may be extended to the plane at infinity, on which they act as M\"obius transformations. Every M\"obius transformation of ${\mathbb{S}}^2$ corresponds uniquely to a symmetry of hyperbolic space.

\subsection{Circle packing}

In the form that we need it, the Koebe--Thurston--Andreev circle packing theorem~\cite{Ste-ICP-05} states that the vertices of every maximal planar graph may be represented by circles with disjoint interiors, such that two vertices are adjacent if and only if the corresponding two circles are tangent.

It is not known how to find circle packings in strongly polynomial time, but it is possible to find approximate packings by numerical algorithms that are polynomial in both the number of vertices and the accuracy of approximation~\cite{Moh-DM-93}.
\ifFull
\else
We use a numerical relaxation procedure by Collins and Stephenson~\cite{ColSte-CGTA-03} that
\fi
repeatedly chooses a circle in the packing, computes the angular defect by which its neighboring circles either fail to surround it or surround it by a larger angle than $2\pi$, determines a representative radius for the neighbors such that if they all had the same radius as each other they would have the same defect, and replaces the radius of the central circle by a new value such that, if the surrounding circles all had the representative radius, the angular defect would be zero. Each of these steps may be performed using simple trigonometric calculations. As Collins and Stephenson show, this method converges rapidly to a unique solution, the system of radii for a valid packing. Once a close approximation to the the radii has been calculated, the positions of the circle centers are not difficult to determine.

\begin{figure}[t]
\centering\includegraphics[height=1.25in]{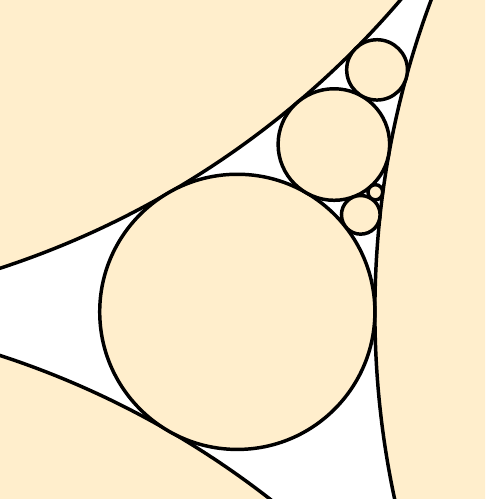}\qquad\qquad
\includegraphics[height=1.25in]{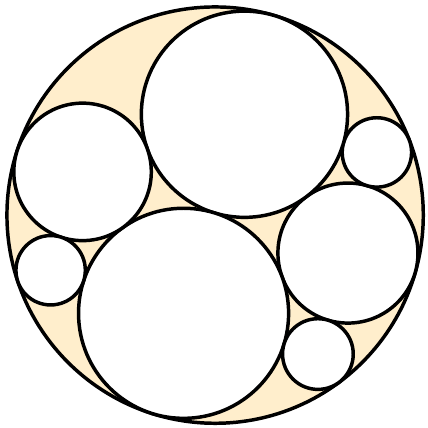}
\caption{Left: the central region of a circle packing constructed by the Collins--Stephenson procedure; the outer three circles, shown only partially in the figure, have equal radii. Right: a transformed packing with one circle exterior to the others, maximizing the minimum circle radius.}
\label{fig:optmob}
\end{figure}

The circle packings constructed by the Collins--Stephenson procedure (with equal outer radii) yield unsatisfactory Lombardi drawings, with one vertex placed at $\infty$. To improve our drawings, we find a M\"obius transformation of the packing for which one chosen circle surrounds all the others. Among all  transformations fixing the outer circle we choose the one that maximizes the radius of the smallest inner circle (Figure~\ref{fig:optmob}). This problem of finding a radius-maximizing M\"obius transformation may be expressed (using the connection between M\"obius transformations and three-dimensional hyperbolic geometry) as a \emph{quasiconvex program}, a problem of finding the minimum value of a pointwise maximum of quasiconvex functions. It may be solved either combinatorially in linear time using LP-type optimization algorithms or numerically using local improvement procedures;  the theory of quasiconvex programs  guarantees that there are no local optima in which the local improvement might get stuck~\cite{BerEpp-WADS-01,Epp-MSRI-05}. Since we are already using a numerical method to find circle packings, our implementation also takes the numerical approach to find the best transformation.

\subsection{Triangle centers}

A \emph{triangle center} is a function from Euclidean triangles to Euclidean points
\ifFull
\else
such that
\fi
performing a similarity and then constructing the center gives the same result as constructing the center and then performing the similarity. Hundreds of triangle centers are known, and include many well-known points determined from a triangle, such as its centroid, circumcenter, incenter, and orthocenter~\cite{Kim-TCCT-98}.
It is convenient, in computing the position of a triangle center, to use \emph{barycentric coordinates}, weights for which the center is the weighted average of the vertices. 

There are two triangle centers that (as an unordered pair of points) are equivariant under M\"obius  transformations,  not just  Euclidean similarities. One of these two centers, the first \emph{isodynamic point}, may be constructed by transforming the given triangle to an equilateral triangle (in such a way that the interiors of the circumcircles of the triangles map to each other), choosing the centroid of the equilateral triangle, and reversing the transformation. The second isodynamic point may be constructed similarly using $\infty$ in place of the centroid.

\section{Cubic polyhedral graphs}

We first explain the simplest case of our Lombardi drawing algorithm, in which the planar graph to be drawn is 3-connected and 3-regular.
We will later simplify the algorithm, but in terms of three-dimensional hyperbolic geometry, our drawings may be constructed by the following steps:
\begin{itemize}
\item Construct the dual graph of the given input graph, a maximal planar graph, and its (unique) planar embedding.
\item Apply the Collins--Stephenson procedure to realize the dual maximal planar graph as the intersection graph of a collection of tangent circles $C_i$.
\item Use quasiconvex programming to find a M\"obius transformation of the circles taking them to a configuration in which one circle $C_0$ is exterior to all the others, and maximizing the minimum radius of the internal circles.
\item Let the plane on which the circles are packed be the boundary of a three-dimensional halfspace model for hyperbolic space. Each circle $C_i$ forms the set of points at infinity for a hyperbolic plane $H_i$ within this space, which in the halfspace model is represented as a hemisphere. With this model, isometries of hyperbolic space correspond to M\"obius transformations of the plane and vice versa.
\item Construct the three-dimensional hyperbolic Voronoi diagram of the hyperbolic planes $H_i$, a partition of hyperbolic space into cells within which every point is closer to $H_i$ (as measured using hyperbolic distance) than to any other one of these planes. The bisector between two hyperbolic planes (the set of points equidistant from both of them) is itself a hyperbolic plane. The cell for $H_i$ is a hyperbolic convex polyhedron, the intersection of hyperbolic halfspaces bounded by these bisectors.
\item Compute the intersection of the Voronoi diagram boundaries with the plane at infinity
(the original Euclidean plane on which we drew the circle packing). The bisectors of the Voronoi diagram form circular arcs in this plane. Triples of bisectors meet at Voronoi edges, which intersect the plane in vertices at which three circular arcs meet. The Voronoi cell for a hyperplane $H_i$ meets the plane at infinity in a two-dimensional region, containing circle $C_i$ (or contained in it for $i=0$) and bounded by circular arcs and vertices. Therefore, the arcs and vertices form a drawing of the dual graph to the intersection graph of circles, which is our initial 3-regular graph.\end{itemize}

\begin{wrapfigure}[15]{r}{0.37\textwidth}
\vskip-4.5ex
\centering\includegraphics[width=.37\textwidth]{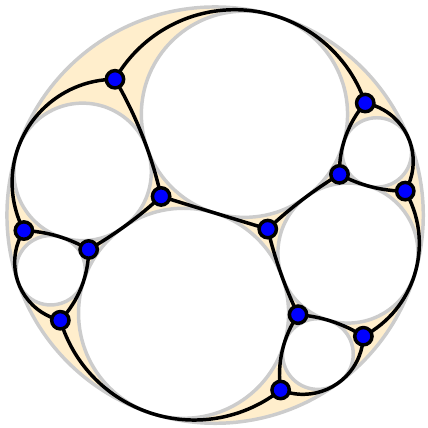}
\vskip-1ex
\caption{Lombardi drawing of the Frucht graph derived from the circle packing of Figure~\ref{fig:optmob}.}
\label{fig:frucht2}
\end{wrapfigure}
To verify that this process produces Lombardi drawings, we must show that the three arcs meeting at each vertex form $120^\circ$ angles. For any three mutually tangent circles of the circle packing, there exists a M\"obius transformation taking their tangencies to an equilateral triangle. In the transformed packing, the bisectors between the three circles are lines, the axes of symmetry of the equilateral triangle, which meet at $120^\circ$ angles. Because M\"obius transformations are
conformal mappings,  the three curves in our un-transformed drawing  have the same $120^\circ$ angles at the vertex where they meet.

Three-dimensional Voronoi diagrams may have quadratic complexity, but  we can shortcut this  bound by constructing our drawing directly from the circle packing.

\begin{theorem}
\label{thm:3con}
If we are given as input an $n$-vertex 3-regular 3-connected planar graph, we can produce a planar Lombardi drawing for the graph in time $T+O(n)$, where $T$ denotes the time needed to find and optimally transform a circle packing dual to the graph.
\end{theorem}

\begin{proof}
Each vertex of the drawing lies within the cusp formed by three mutually tangent circles.
By M\"obius invariance, it is the isodynamic point of the triangle formed by the three tangency points of its cusp. We may calculate its position as weighted average of the three tangency points, using the known barycentric coordinates of the isodynamic point as weights. Each circular arc of the drawing has two vertices as its endpoints and passes through one tangency point, which together determine its location.
\end{proof}

\section{Two-connected subcubic graphs}

We next describe how to extend our Lombardi drawing technique to 2-connected graphs. A 2-connected graph may have vertices of degree two or three, but the degree-2 vertices may be suppressed, forming a 3-regular multigraph with the same connectivity. If the multigraph has a planar Lombardi drawing, so does the original graph with the degree-2 vertices, as the vertices may be restored by subdividing edges without changing the Lombardi property. Therefore, for the most part within this section we assume that the given graph remains 3-regular.

We use a standard tool for decomposing 2-connected graphs, the SPQR tree~\cite{DiBTam-FOCS-89,DiBTam-ICALP-90,GutMut-GD-01,HopTar-SJC-73,Mac-DMJ-37}. An SPQR tree for a graph~$G$ is a tree structure in which each tree node is associated with a graph $C_i$, known as a 3-connected component of $G$. In each 3-connected component, some of the edges are labeled as ``virtual'', and each edge of the SPQR-tree is labeled by a pair of oriented virtual edges from its two endpoints; each virtual edge of a component is associated in this way with exactly one tree edge. The given graph $G$ may be formed by gluing the components $C_i$ together, by identifying pairs of endpoints of virtual edges according to the labeling of the tree edges, and then deleting the virtual edges themselves. The nodes of an SPQR tree have three types: R nodes, in which the associated graph is 3-connected, S nodes, in which the associated graph is a cycle, and P nodes, in which the associated graph is a bond graph, a multigraph with two vertices and three or more parallel edges. With the additional constraint that no two S nodes and no two P nodes may be adjacent, the SPQR tree is uniquely determined from $G$, and may be constructed in linear time.

The SPQR trees of 3-regular graphs have an additional structure that will be helpful for us:

\begin{lemma}[Pootheri~\cite{Poo-SE-01}, Eppstein and Mumford~\cite{EppMum-SoCG-10}]
\label{lem:spqr3}
A 2-connected graph $G$ is 3-regular if and only if each edge in its SPQR tree has exactly one S node as an endpoint, each S node is associated with an even cycle that alternates between virtual and non-virtual edges, each P node is associated with a three-edge bond graph, and each R node is associated with a graph that is itself 3-regular.
\end{lemma}

\begin{figure}[t]
\centering\includegraphics[height=1in]{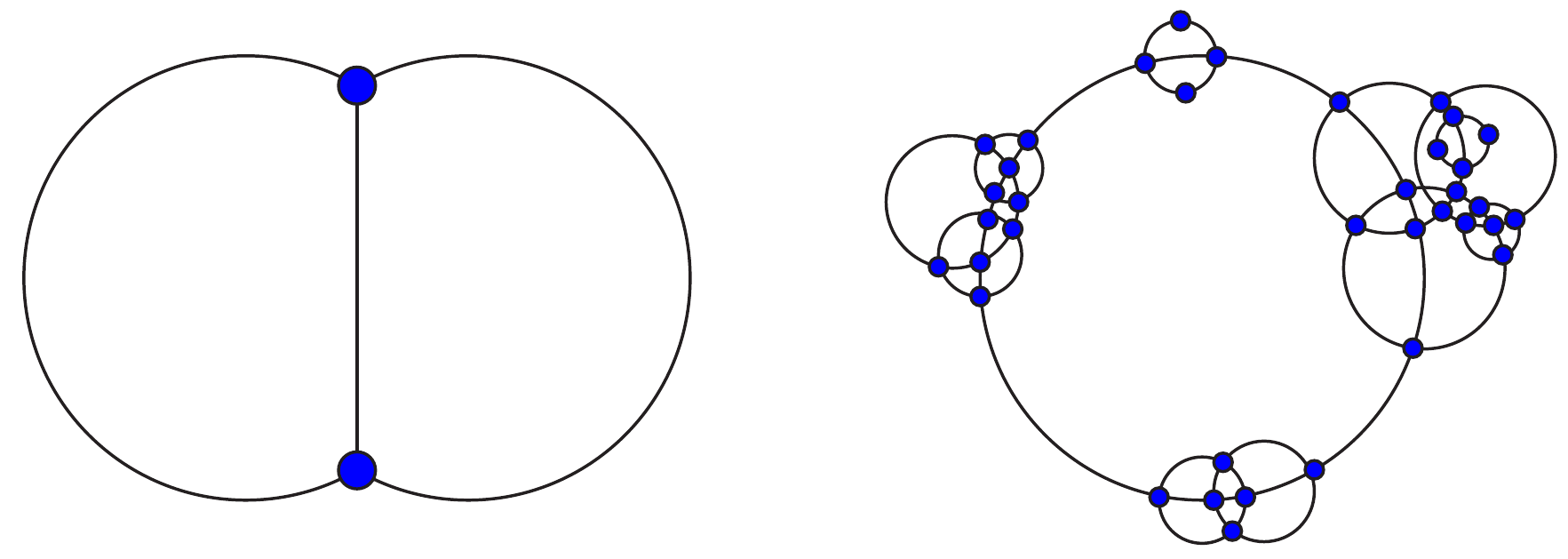}
\caption{Left: planar Lombardi drawing for a 3-edge P node. Right: gluing multiple Lombardi drawings together on an S node with alternating virtual and non-virtual edges (schematic view, not an actual Lombardi drawing).}
\label{fig:spqr-glue}
\end{figure}

\begin{theorem}
\label{thm:2con}
If we are given as input an $n$-vertex 2-connected planar graph with maximum degree~3, we can produce a planar Lombardi drawing for the graph in time $T+O(n)$, where $T$ denotes the time needed to find and optimally transform a family of circle packings with total cardinality $O(n)$.
\end{theorem}

\begin{proof}
Suppress the degree-two vertices of the given graph to produce a 3-regular 2-connected graph, decompose it into an SPQR tree with the additional structure of Lemma~\ref{lem:spqr3}, and use Theorem~\ref{thm:3con} to construct a planar Lombardi drawing of the 3-connected component associated with each $R$ node (Figure~\ref{fig:spqr-glue}, left). In each S node of the tree, glue the drawings from the adjacent tree nodes together by using
\ifFull
\else
inversions centered near their virtual edges
\fi
to expand the virtual edges, move these edges to the outer face of their drawings, and shrink the rest of the drawings, and then aligning the circular arcs representing their virtual edges so that they all lie on a common circle (Figure~\ref{fig:spqr-glue}, right). Finally, subdivide the edges of the drawing as necessary to restore the suppressed degree-two vertices.
\end{proof}

\section{Subcubic graphs with bridges}

\begin{wrapfigure}[12]{r}{0.5\textwidth}
\vskip-5ex
\centering\includegraphics[width=0.5\textwidth]{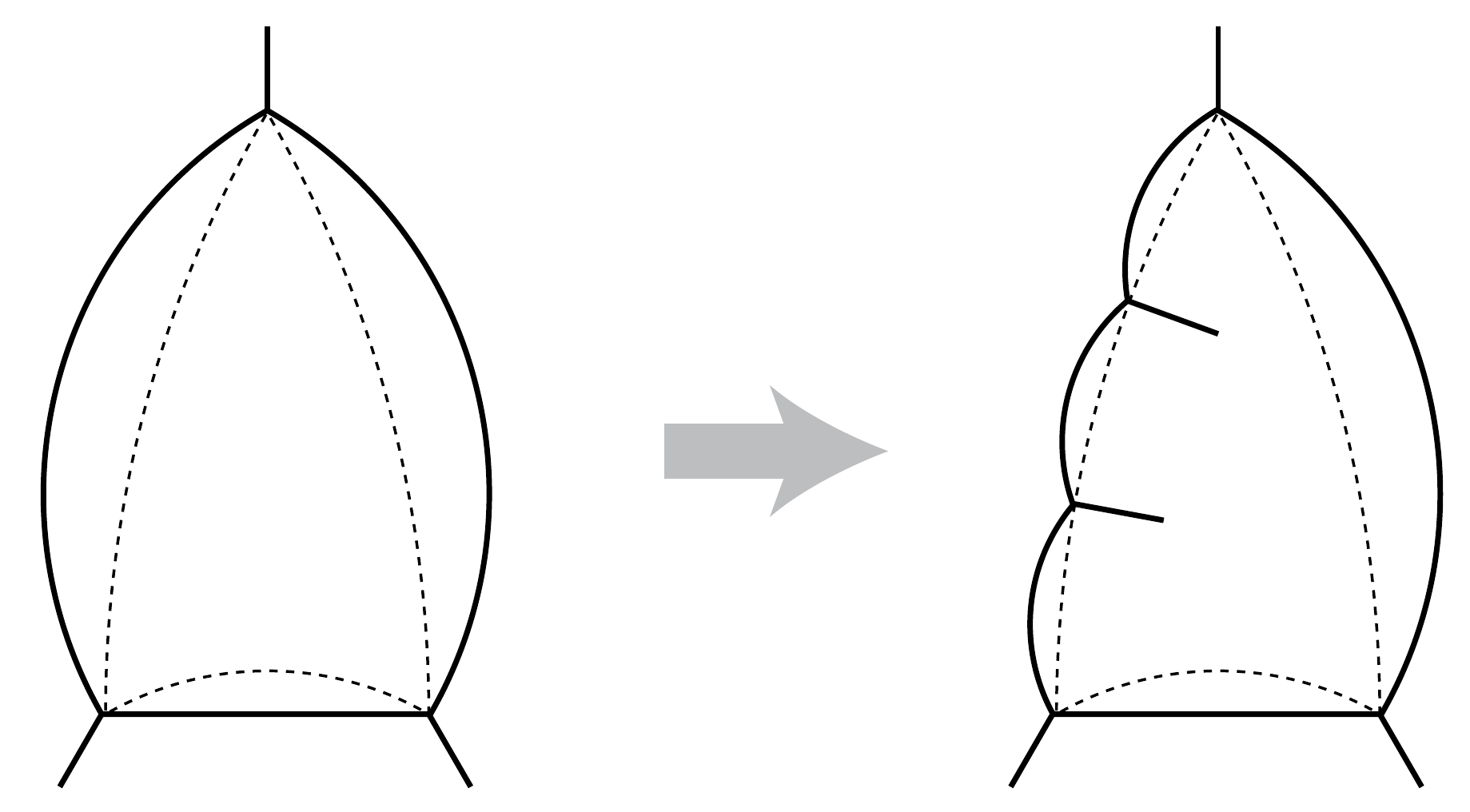}
\vskip-1ex
\caption{Modifying a 3-connected Lombardi drawing to attach bridges to subdivision points along an edge.}
\label{fig:make-bridges}
\end{wrapfigure}
We are finally ready to describe our algorithm for constructing planar Lombardi drawings for arbitrary subcubic planar graphs, possibly including bridge edges.

Our algorithm begins by deleting all the bridge edges from the graph, and suppressing all degree two vertices, leaving a collection of isolated vertices and 2-connected 3-regular subgraphs. It then constructs the SPQR tree of each 2-connected subgraph, forming its 3-connected components, and applies Theorem~\ref{thm:3con} to construct a planar Lombardi drawing for each R node in each SPQR tree.

Within each face of each of these drawings, form a sequence of circular arcs forming $30^\circ$ angles to the arcs of the face and $60^\circ$ angles to each other, forming a smaller inset polygon (Figure~\ref{fig:make-bridges}, left). These arcs cannot intersect
each other -- this is straightforward to see within each triangular cusp of the circle packing
from which the drawing was constructed, and within each circle of the circle packing (viewing the circle as a disk model for the hyperbolic plane) the new arcs are confined to disjoint halfspaces.
For each edge $e$ of one of these
drawings that corresponds to a path of suppressed degree-two vertices some of which were bridge endpoints, let $k$ be the number of bridge endpoints of $e$. Choose arbitrarily one or the other of the two circular arcs $A$ forming $30^\circ$ angles to $e$, and replace $e$ by a sequence of $k+1$ circular arcs, all of which form the same $30^\circ$ angle to $A$. These arcs meet each other at $120^\circ$ angles, and we may extend a bridge from each of the vertices formed in this way to a new degree-one vertex (Figure~\ref{fig:make-bridges}, right). By a similar construction we unsuppress and attach bridge edges to the suppressed bridge endpoints within the P nodes and S nodes of the SPQR tree.

Next, we glue the components within each SPQR tree together, as in Theorem~\ref{thm:2con}. At the same time, for each isolated vertex created by the bridge deletion step, we create a drawing of a claw $K_{1,3}$ consisting of three unit-length line segments meeting at $120^\circ$ angles. Additionally, we unsuppress the remaining degree-two vertices of the original graph, by subdividing edges of these drawings. After this step, we have separate drawings for each block (biconnected component) of the original graph, in which each bridge incident to the block has a degree-one vertex at its other end.

Finally, for each two blocks that should be connected by a bridge, we apply an inversion centered on the degree-one endpoint of the bridge, causing the bridge edge to become an infinite ray exterior to the rest of the drawing. After this transformation the two blocks may be glued together by making their two copies of the bridge edge lie on the same line.

This completes the proof of our main result:

\begin{theorem}
\label{thm:1con}
If we are given as input an $n$-vertex planar graph with maximum degree~3, we can produce a planar Lombardi drawing for the graph in time $T+O(n)$, where $T$ denotes the time needed to find and optimally transform a family of circle packings with total cardinality $O(n)$.
\end{theorem}

\section{Four-regular graphs}

\begin{figure}[t]
\centering\includegraphics[height=1.25in]{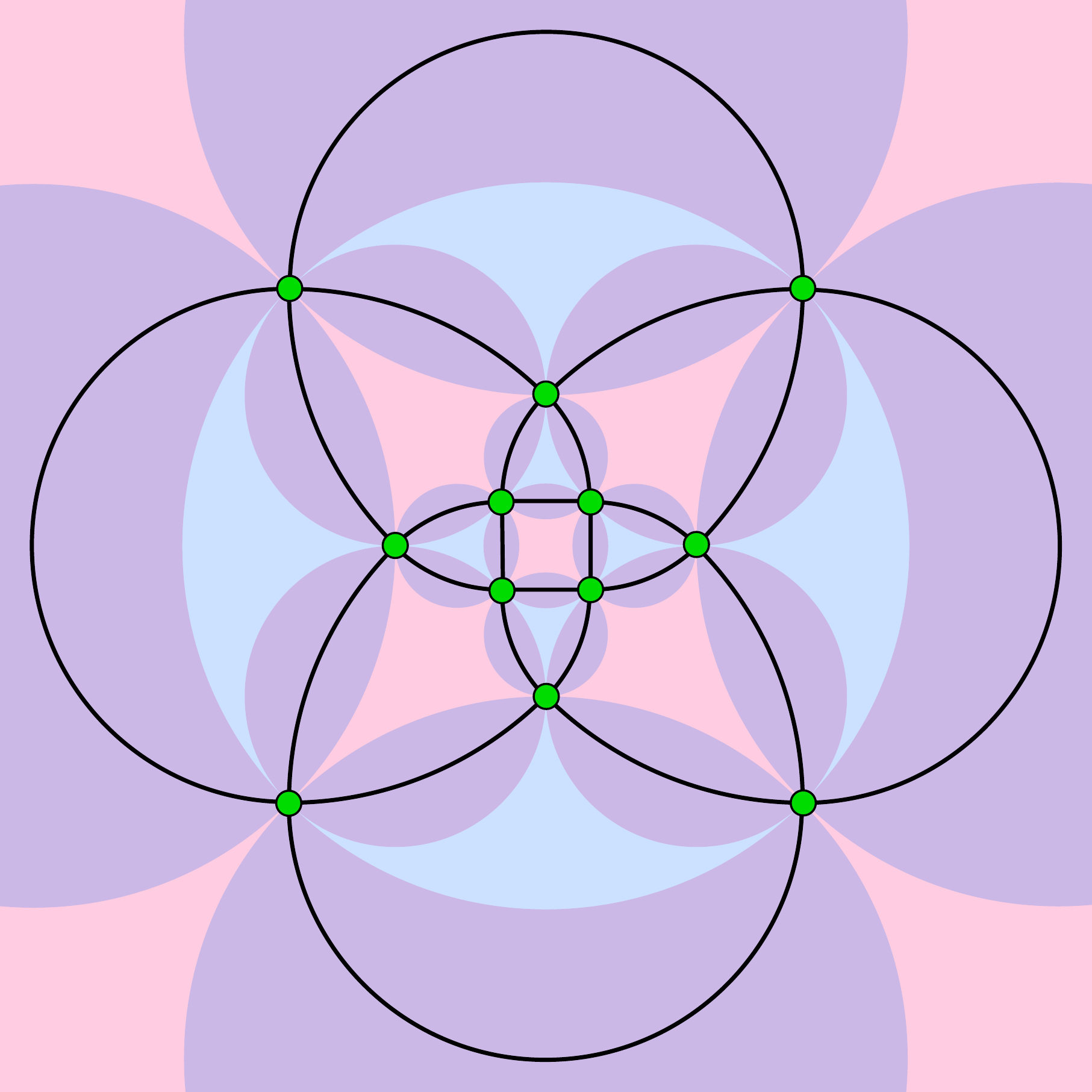}\qquad
\includegraphics[height=1.25in]{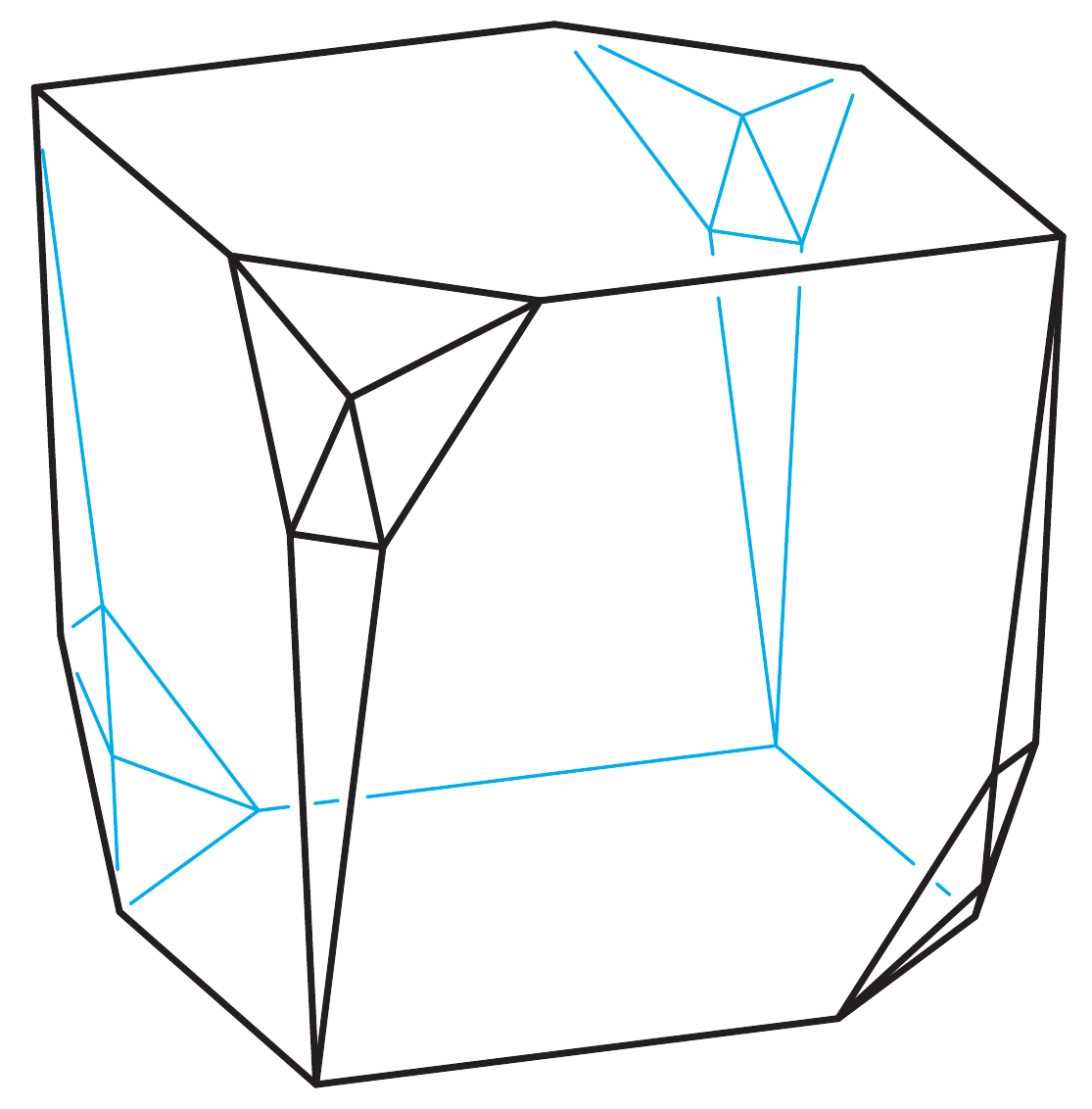}\quad
\includegraphics[height=1.25in]{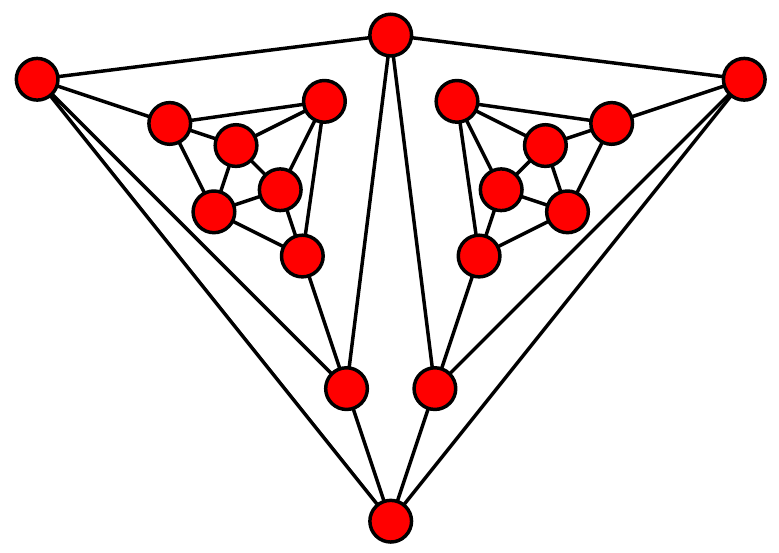}
\caption{Left: Dual circle packings for a cube and octahedron, with a Lombardi drawing of their medial graph, the cuboctahedron. Center: A 4-regular polyhedron that is the medial graph of a non-polyhedral graph, described by Dillencourt and Eppstein~\cite{EppDil-EGM-03}. Right: A 4-regular planar graph with no Lombardi drawing.}
\label{fig:medial}
\end{figure}

An alternative form of the circle packing theorem states every 3-connected planar graph $G$ and its dual can be simultaneously represented by tangent circles, so that two circles from the two packings are orthogonal if they represent a vertex $v$ of $G$ and a face of the $G$ containing $v$, and disjoint otherwise. The packing is unique up to M\"obius transformation, and may be found by a numerical procedure similar to the one for circle packing of maximal planar graphs~\cite{Moh-DM-93}.

These dual packings allow us to construct a planar Lombardi drawing of the \emph{medial graph} of $G$, the 4-regular graph formed by placing a vertex on the midpoint of each edge of $G$ and connecting two of these new vertices by an edge whenever they are the midpoints of consecutive edges on the same face. Each pair of orthogonal circles in the packing and dual packing of $G$ form a lune, which may be bisected by a circular arc forming a $45^\circ$ angle with both circles. Two circular arcs meeting at the same point of tangency form right angles or $180^\circ$ angles to each other, so the collecting of bisecting arcs forms a drawing of the graph with one vertex for each point of tangency and one edge for each pair of orthogonal circles, which is the medial graph of $G$ (Figure~\ref{fig:medial}, left). As before, this drawing has an interpretation as the intersection of a three-dimensional hyperbolic Voronoi diagram with the plane at infinity, and as before it can be extended to certain 2-connected graphs using the SPQR tree. However, this technique does not apply to 3-connected 4-regular planar graphs that are not medial graphs of polyhedral graphs, such as the one shown in Figure~\ref{fig:medial}, center.

The 4-regular 2-connected planar graph in Figure~\ref{fig:medial}, right has no Lombardi drawing. We defer the proof to an appendix.

\section{Implementation}

\begin{figure}[p]
  \centering
  \subfloat[Markstr\"om's 24-vertex graph with no 4- or 8-cycles~\cite{Mar-CN-04}]{\includegraphics[width=0.3\textwidth]{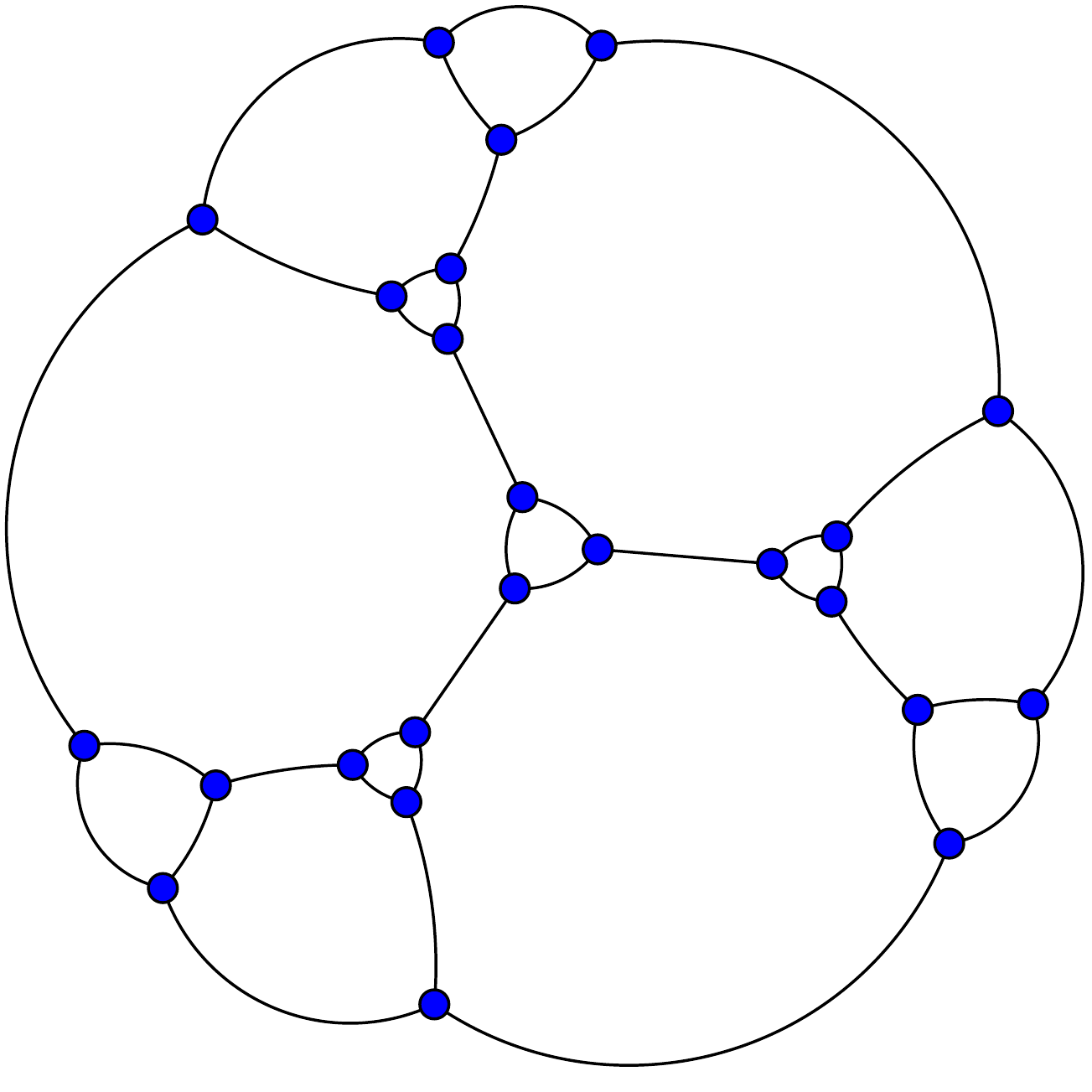}}
  \hfil  \subfloat[Do{\v{s}}li{\'c}'s 38-vertex graph with girth five and cyclic edge connectivity three~\cite{Dos-JMC-02}]{\includegraphics[width=0.3\textwidth]{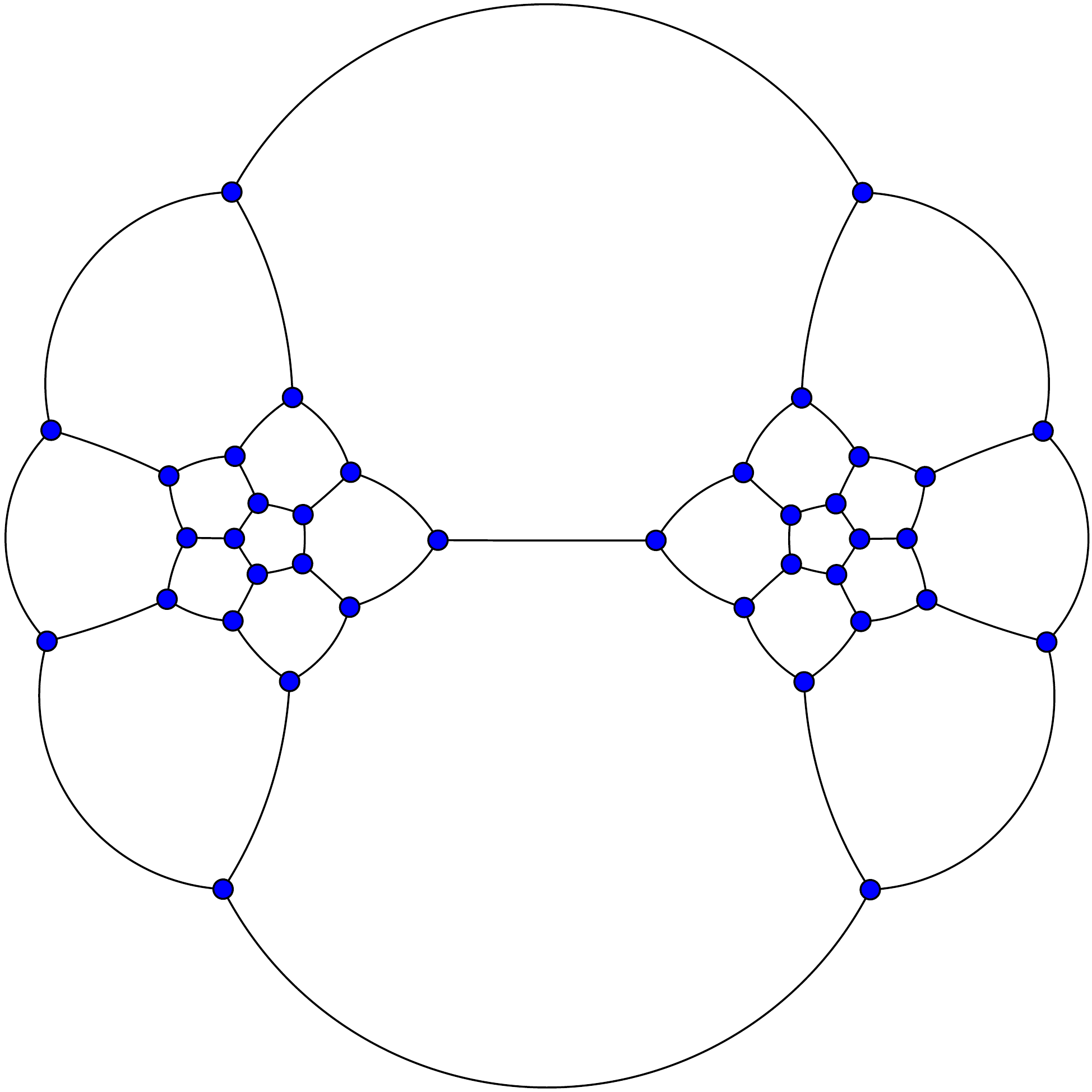}}
  \hfil
  \subfloat[252-vertex hexagonal mesh]{\includegraphics[width=0.3\textwidth]{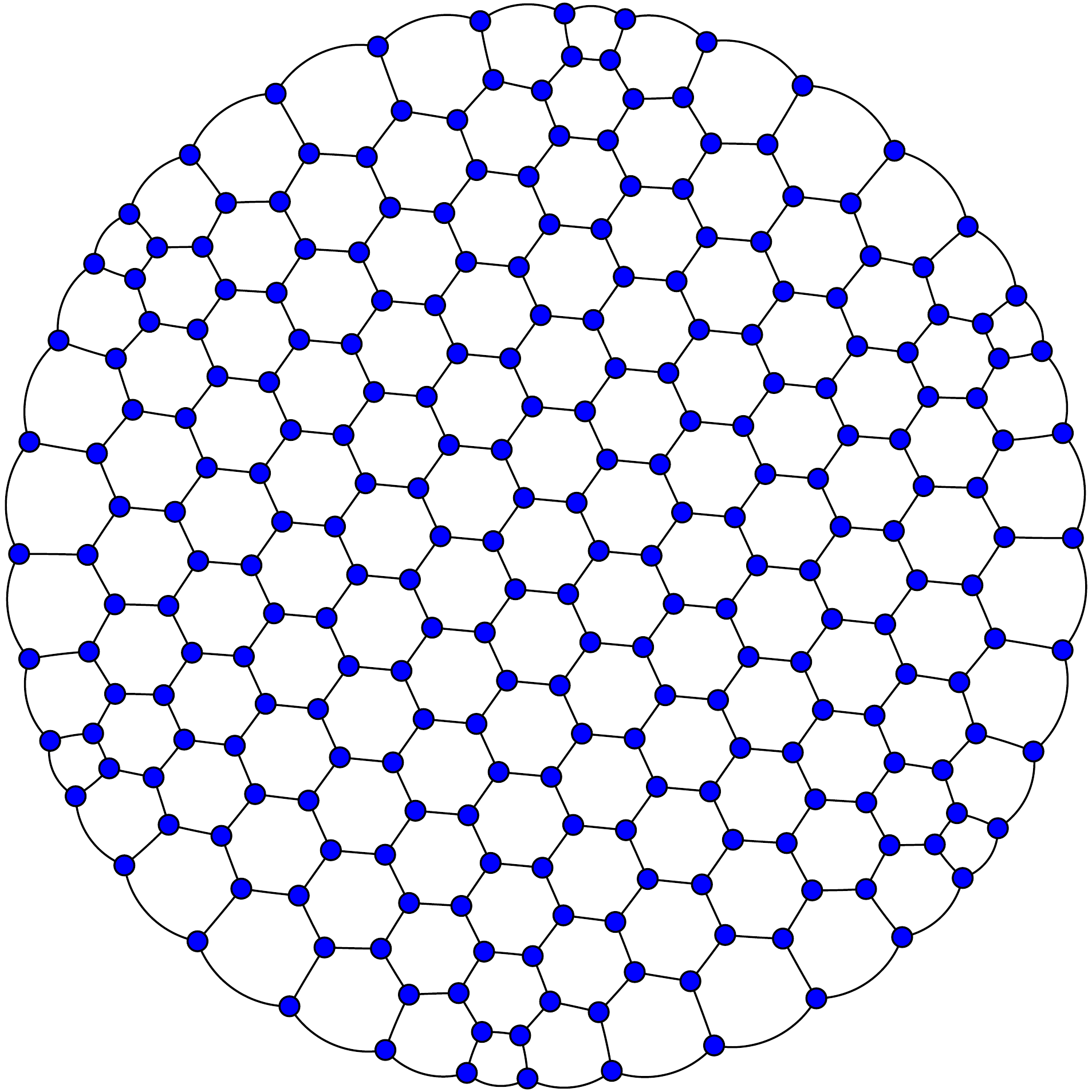}}\\
  \subfloat[Tutte's 46-vertex non-Hamiltonian cubic planar 3-connected graph~\cite{Tut-JLMS-46}]{\includegraphics[width=0.3\textwidth]{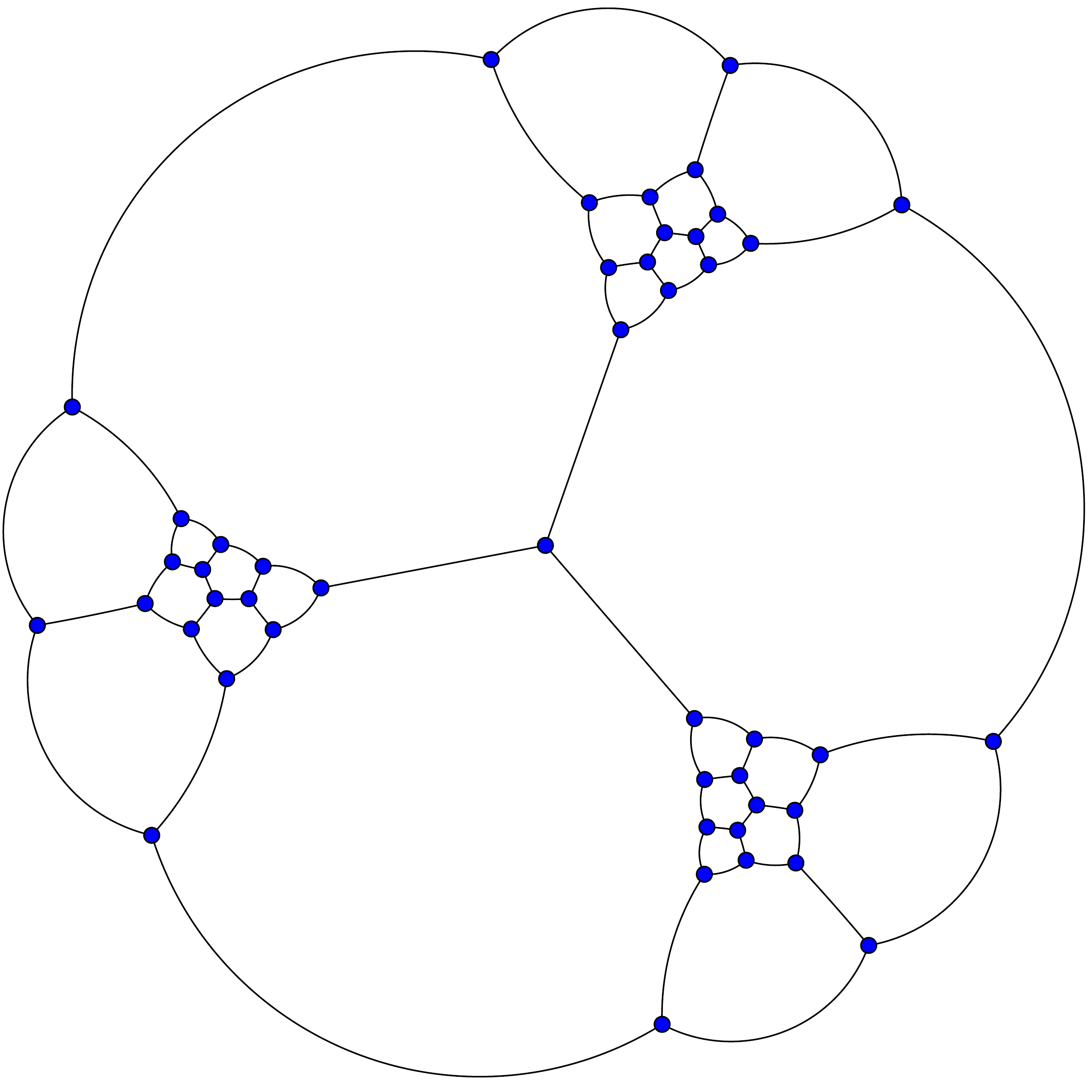}}
  \hfil
  \subfloat[Grinberg's 46-vertex non-Hamiltonian graph with cyclic edge connectivity five~\cite{Gri-LMY4-68}]{\includegraphics[width=0.3\textwidth]{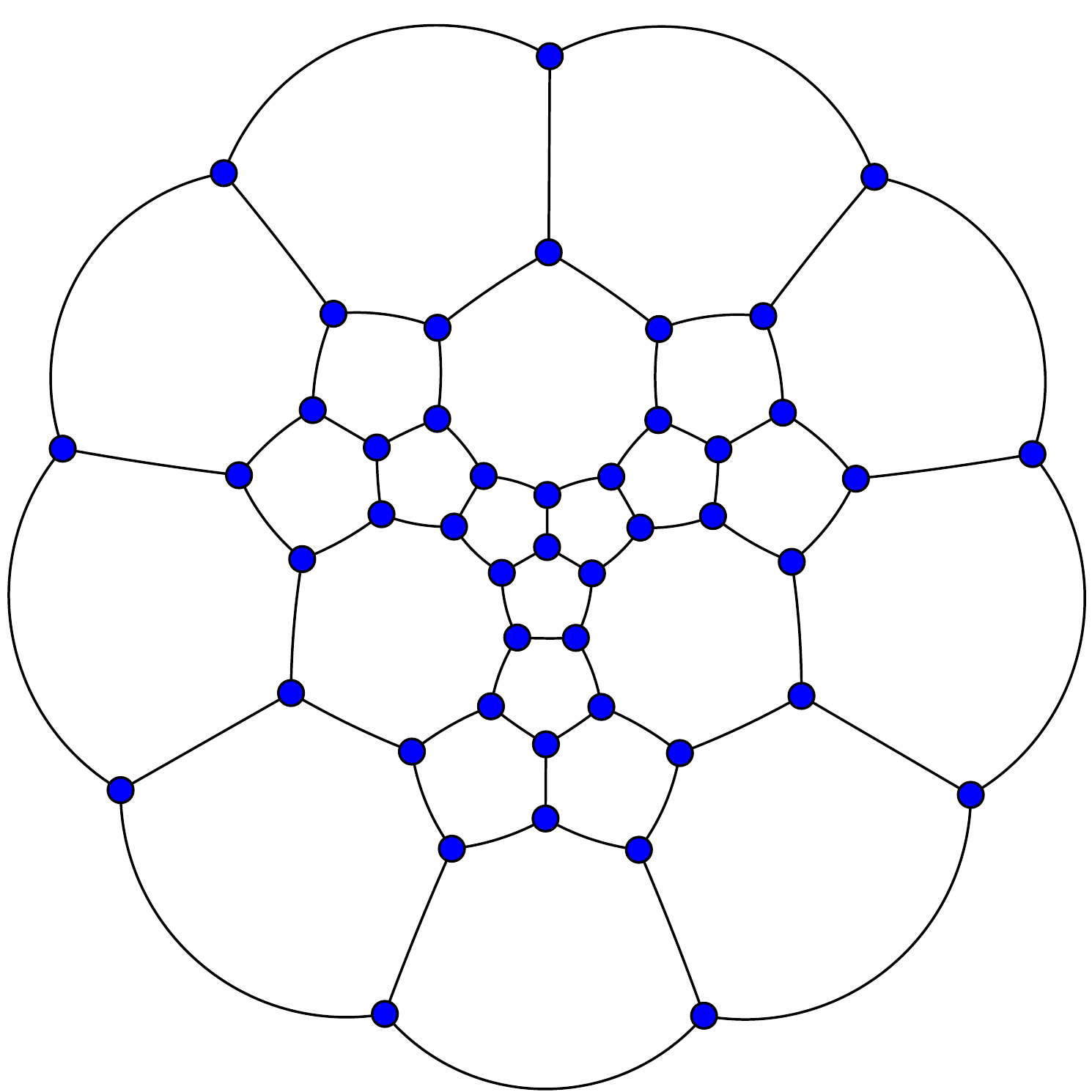}}
  \hfil
  \subfloat[46-vertex Halin graph formed from a complete ternary free tree]{\includegraphics[width=0.3\textwidth]{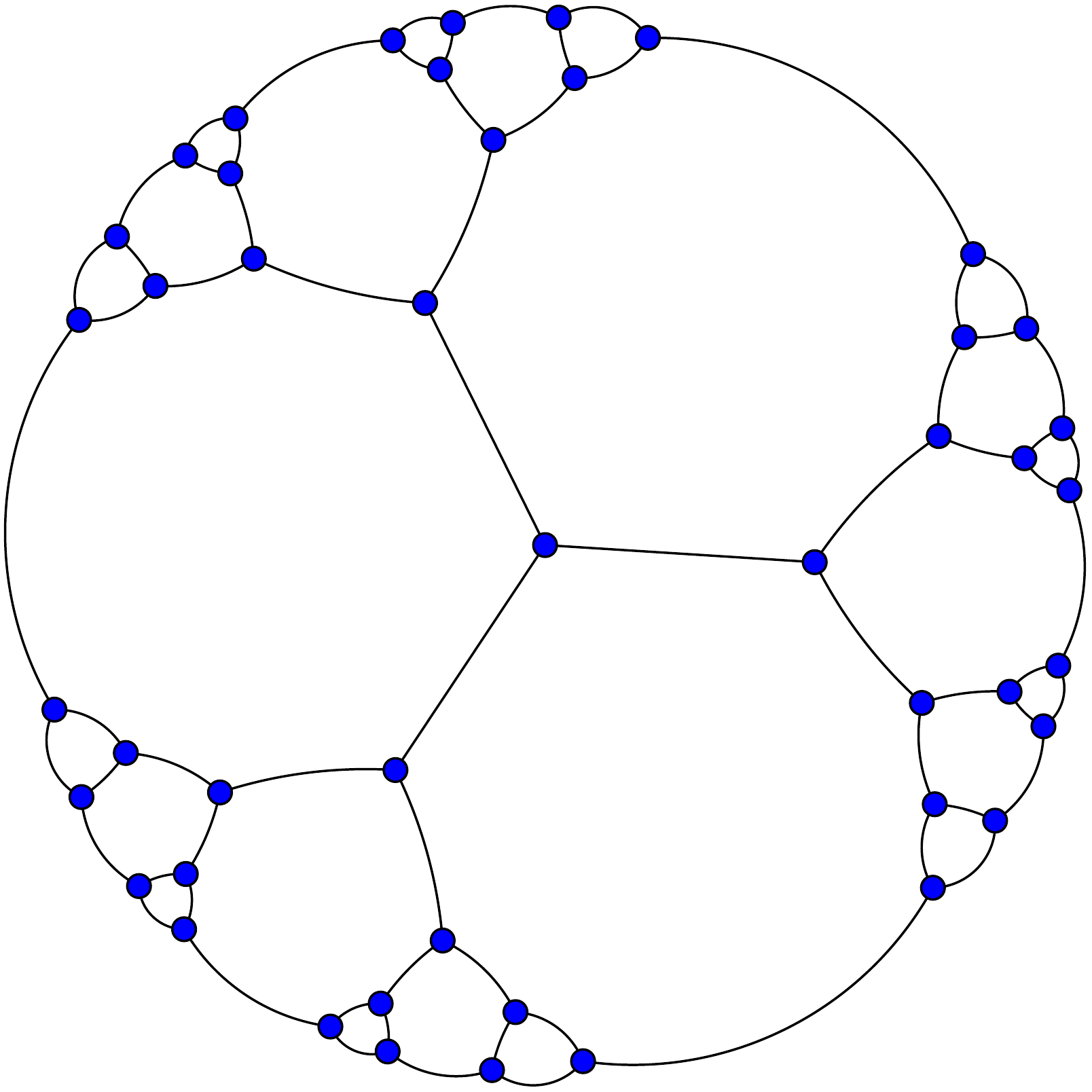}}\\
  \subfloat[60-vertex buckyball or truncated icosahedron]{\includegraphics[width=0.35\textwidth]{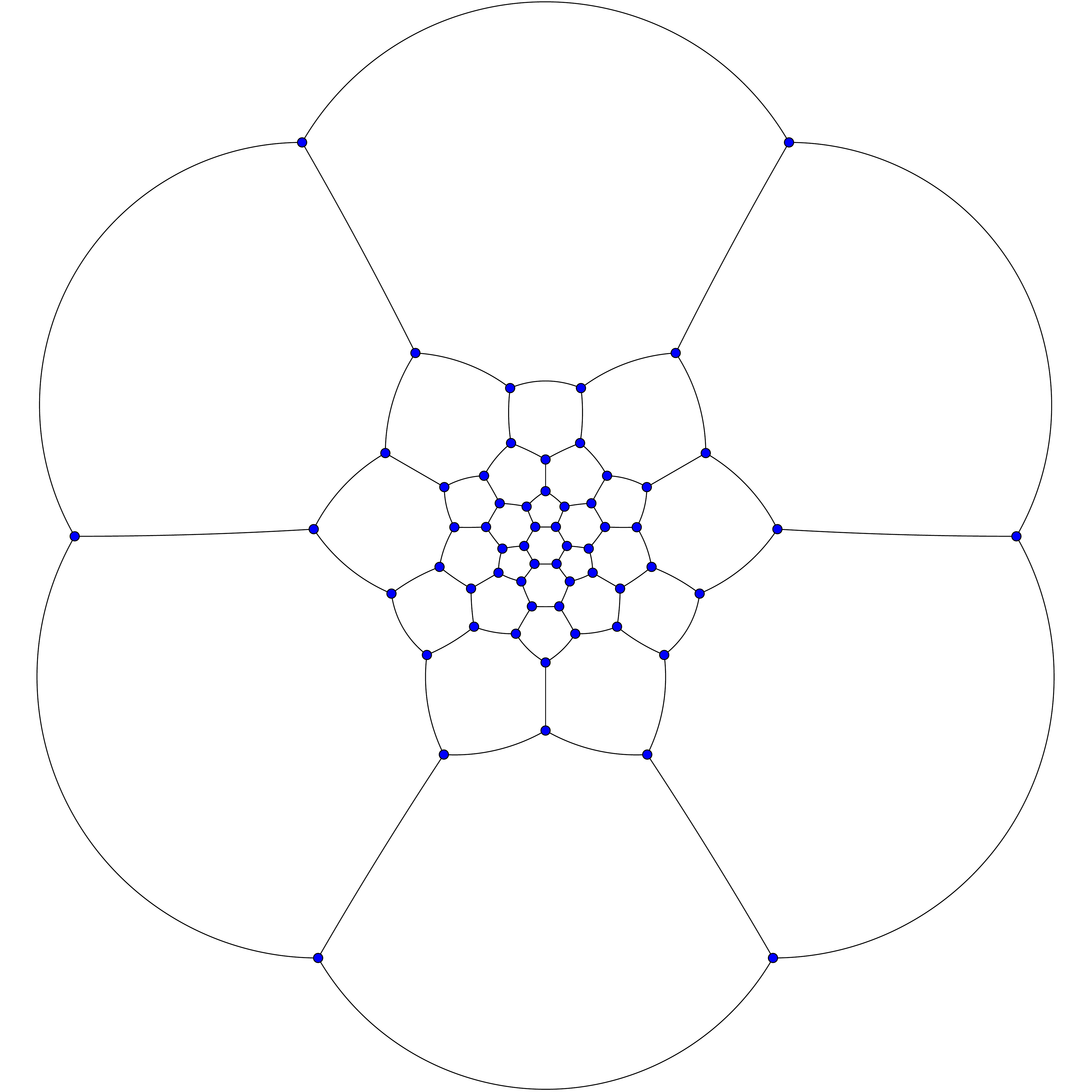}}
  \hfil
  \subfloat[Irregular 69-vertex graph]{\includegraphics[width=0.35\textwidth]{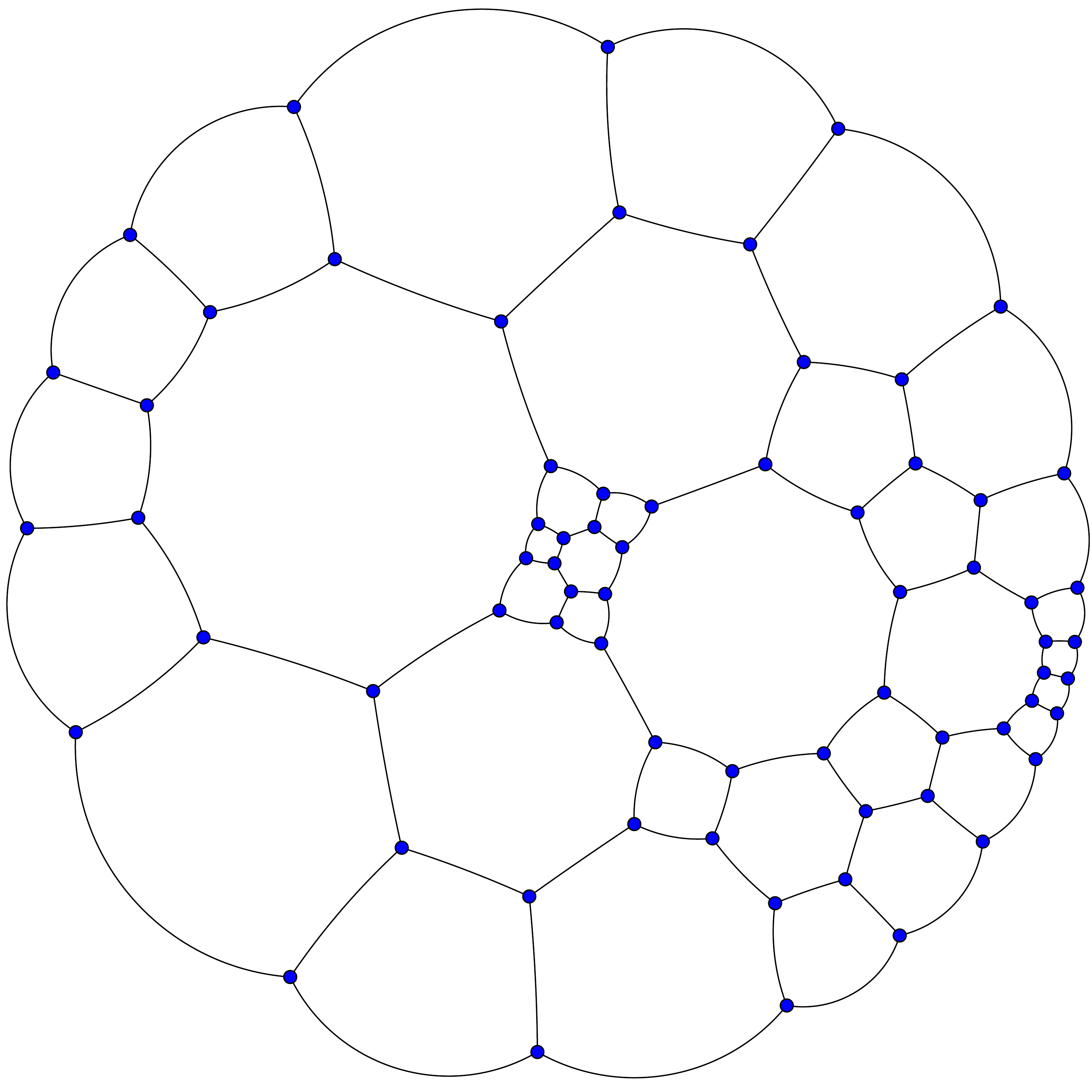}} \\
  \caption{Sample drawings from our implementation.}
  \label{fig:samples}
\end{figure}

We implemented in Python the algorithm for 3-connected graphs, including the Collins--Stephenson circle packing algorithm and a numerical improvement method for finding optimal M\"obius transformations. Our implementation takes as input a text file with one line per vertex; each line lists the identifiers for a vertex and its  three neighbors in clockwise order. The output drawing is represented in the SVG vector graphics file format. Figure~\ref{fig:samples} shows some drawings created by our implementation.

\section{Conclusions}

We have shown that all planar graphs with maximum degree three (and some planar graphs of degree four) have planar Lombardi drawings, greatly extending the classes of graphs for which such drawings are known to exist, and we have implemented our algorithm for the special case of 3-connected 3-regular graphs.

The drawings constructed by our new algorithm have a natural, organic shape, in which the outer face and all the interior faces are approximately circular, resembling soap bubble complexes. If the maximum number of edges per face of the input graph is bounded, then adjacent circles in the circle packing will have radii whose ratio is also bounded, from which it follows that the vertex spacing of the drawing is locally uniform. Our drawings automatically display any global cyclic or dihedral symmetries of the input graph, and may be used even in cases such as the truncated icosahedron which  our previous methods, designed for graphs with a high degree of symmetry~\cite{DunEppGoo-GD-10b}, are incapable of handling. They also display certain local symmetries of the graph, in the sense that, if $G$ is a subgraph of a 3-connected graph that is connected to the rest of the graph by exactly three edges, then every copy of $G$ in one of our drawings can be transformed into every other copy of $G$ by a M\"obius transformation.

In a followup paper, in preparation, we extend these results in several directions. In particular, we show that our Lombardi drawing methods for 3-connected and 2-connected planar graphs generate drawings for which it is possible to assign pressures to the faces in such a way that they obey the physical laws governing the static behavior of soap bubbles, and we use this method to characterize the graphs of planar soap bubbles as being exactly the bridgeless 3-regular planar graphs~\cite{Epp-soap}. (Our method for graphs with bridges does not produce valid soap bubbles.) In addition, we extend to intersecting disks the M\"obius-invariant power diagram defined here for disjoint disks, and we find a distance function in the Euclidean plane for which it is the minimization diagram, giving it an intrinsic definition rather than one relying on hyperbolic geometry.

\subsection*{Acknowledgements}

This work was supported in part by NSF grant
0830403 and by the Office of Naval Research under grant
N00014-08-1-1015.

\ifFull
\else
{\raggedright
\bibliographystyle{abuser}
\bibliography{lombardi}}
\fi

\newpage
\appendix
\section{A 4-regular graph with no planar Lombardi drawing}

\begin{figure}[b]
\centering\includegraphics[height=1.5in]{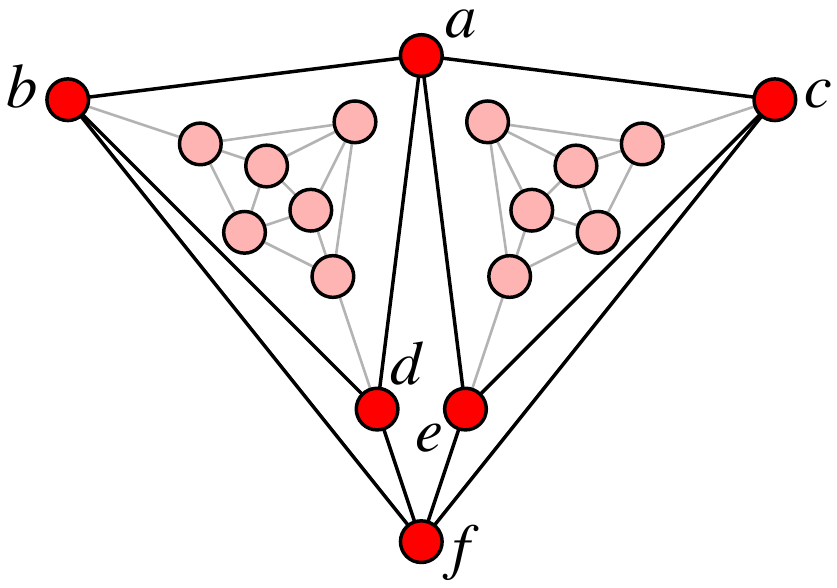}
\caption{The graph $G_{18}$ of Figure~\ref{fig:medial} (right), with vertex labels.}
\label{fig:nonlom-labeled}
\end{figure}

The graph $G_{18}$ of figure~\ref{fig:medial} (right) is shown again in Figure~\ref{fig:nonlom-labeled}. The important parts of this graph, for our purposes, are the six labeled vertices; the remaining vertices are present only to make the graph 4-regular, and could be contracted to single vertices to produce a simpler eight-vertex counterexample to the existence of planar Lombardi drawings while sacrificing regularity. Note that (up to automorphisms) the labeled part of $G_{18}$ has a unique planar embedding, and that the two quadrilateral faces $abfc$ and $adfe$ have their two opposite vertices $a$ and $f$ in common. Because $G_{18}$ is 4-regular, it has a Lombardi drawing with a circular layout, with the arcs representing its edges forming $45^\circ$ angles to the circle on which the vertices are drawn~\cite{DunEppGoo-GD-10b}. However, a circular Lombardi drawing of this type is not guaranteed to be planar, and we will show that $G_{18}$ does not have a planar Lombardi drawing.

A key idea in the proof is to relate quadrilateral faces of any possible drawing, such as faces $abfc$ and $adfe$, with circles, as the following lemma shows.

\begin{lemma}
\label{lem:quad-circumcircle}
Let $pqrs$ be a quadrilateral face of a planar Lombardi drawing, and suppose that all four of its vertices have degree four. Then there exists a circle through all four of the points $p$, $q$, $r$, and $s$, that does not cross any edges of $pqrs$.
\end{lemma}

\begin{figure}[t]
\centering\includegraphics[scale=0.45]{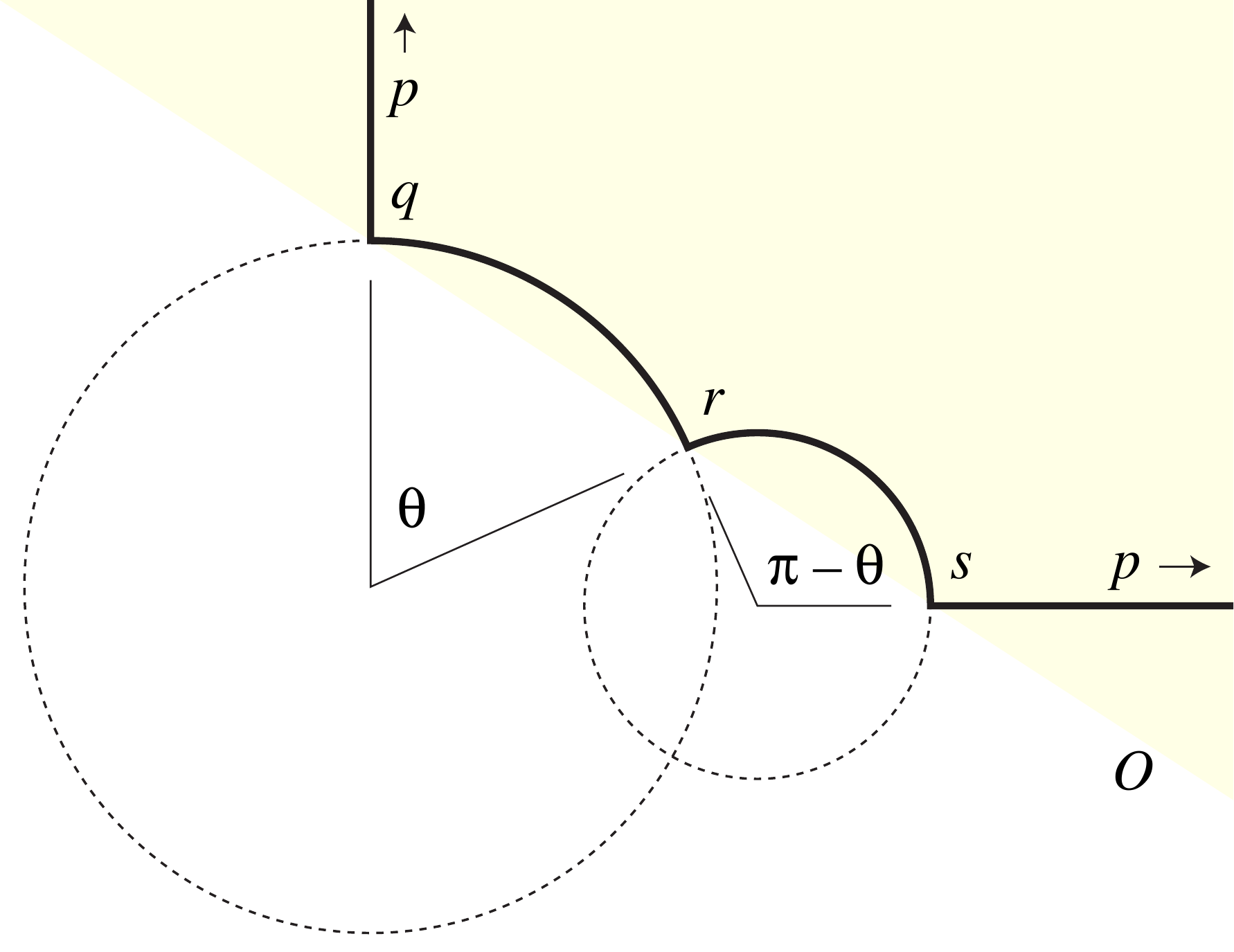}
\caption{Figure for proof of Lemma~\ref{lem:quad-circumcircle}}
\label{fig:quad-circumcircle}
\end{figure}

\begin{proof}
Perform an inversion through a circle centered at $p$, sending $p$ to $\infty$. The transformed copy of $pqrs$ includes two perpendicular rays $qp$ and $sp$, which we may assume without loss of generality are axis-parallel, as shown in Figure~\ref{fig:quad-circumcircle}.

If both $qr$ and $rs$ are drawn as curved arcs (rather than straight line segments) in the transformed drawing, let $\theta$ be the angle $qr$ subtends on its circle; in order for all the angles of the figure to be right, $rs$ must subtend the complementary angle $\pi-\theta$. The angle between arc $qr$ and line segment $qr$ (at either of the two intersection points of these two curves) is $theta/2$, and symmetrically the angle between arc $rs$ and line segment $rs$ is $(\pi-\theta)/2$. Angle $qrs$ is equal to the sum of these two angles together with the right angle between the two arcs; that is, it is $\theta/2 + \pi/2 + (\pi-\theta)/2 = \pi$. Thus, $q$, $r$, and $s$ are collinear.

If one of $qr$ and $rs$ (say $qr$) is not a curved arc, but a straight line segment, it must be axis-parallel in order to form a right angle at $q$, and the other arc $rs$ must be a semicircle in order to form right angles at $r$ and $s$. Again, in this case, $q$, $r$, and $s$ are collinear.

In either case, the line $O$ through the three points $q$, $r$, and $s$ also passes through $p=\infty$ (because every line passes through $\infty$). Reversing the inversion we initially performed, $O$ becomes the desired circle. It can't cross the circular arcs of $pqrs$ because two circles can only intersect at two points and each of these arcs already intersects $O$ at its two endpoints.
\end{proof}

\begin{theorem}
\label{thm:nonlom}
Graph $G_{18}$ has no planar Lombardi drawing.
\end{theorem}

\begin{proof}
If $G_{18}$ had a planar Lombardi drawing, then by Lemma~\ref{lem:quad-circumcircle} its quadrilateral faces $adfe$ and $abfc$ would have circumscribing circles. If they do not coincide, these circles intersect at the two vertices $a$ and $f$; the other four vertices must be placed between them on the four circular arcs connecting these two intersection points.
By performing an inversion centered within face $abfc$ we may assume without loss of generality that the drawing resembles the schematic layout shown in Figure~\ref{fig:nonlom-impossibility}, in the sense that face $abfc$ lies outside of its circle. The other circle, $adfe$, must contain one but not both of the vertices $b$ and $c$; without loss of generality we assume it contains $c$ and does not contain $b$, because the other case is symmetric under a relabeling of the vertices. Then (in order for the drawing to be planar) vertex $d$ may be placed anywhere along the arc of circle $adfe$ closest to $b$. The ray from $c$ radially outwards from the center of circle $abfc$ lies entirely within face $abfc$, and crosses circle $adfe$; vertex $e$ may be placed on its arc either above or below this ray, but by symmetry (flipping the diagram and relabeling the vertices if necessary) we assume that it is above the ray, as depicted schematically in the figure.

For this placement, the edge $ce$ must exit vertex $c$ towards the interior of circle $abfc$, because it forms a $180^\circ$ angle with either edge $ac$ or edge $fc$, both of which are exterior to the circle. It must then cross circle $abfc$ again, turning clockwise as it does, in order to stay within quadrilateral $aefc$ and reach vertex $e$. By a symmetric argument it must exit vertex $e$ towards the exterior of circle $adfe$, because it forms a $180^\circ$ angle with either edge $ae$ or edge $fe$, both of which are interior to the circle. Thus, coming from $c$, it must cross circle $adfe$, turning counterclockwise as it does, to reach the circle a second time.
But in a Lombardi drawing, edge $ce$ must be drawn as a circular arc, which cannot turn clockwise for part of its length and counterclockwise for part of its length. Therefore, this placement is impossible.

The remaining case arises when circles $abfc$ and $adfe$ coincide. But in this case, by similar reasoning, edge $ce$ must exit $c$ towards the interior of the single circle, cross the circle, and reach $e$ from the exterior, an impossibility since a circular arc cannot have three intersection points with a circle.
\end{proof}

\begin{figure}[t]
\centering\includegraphics[scale=0.45]{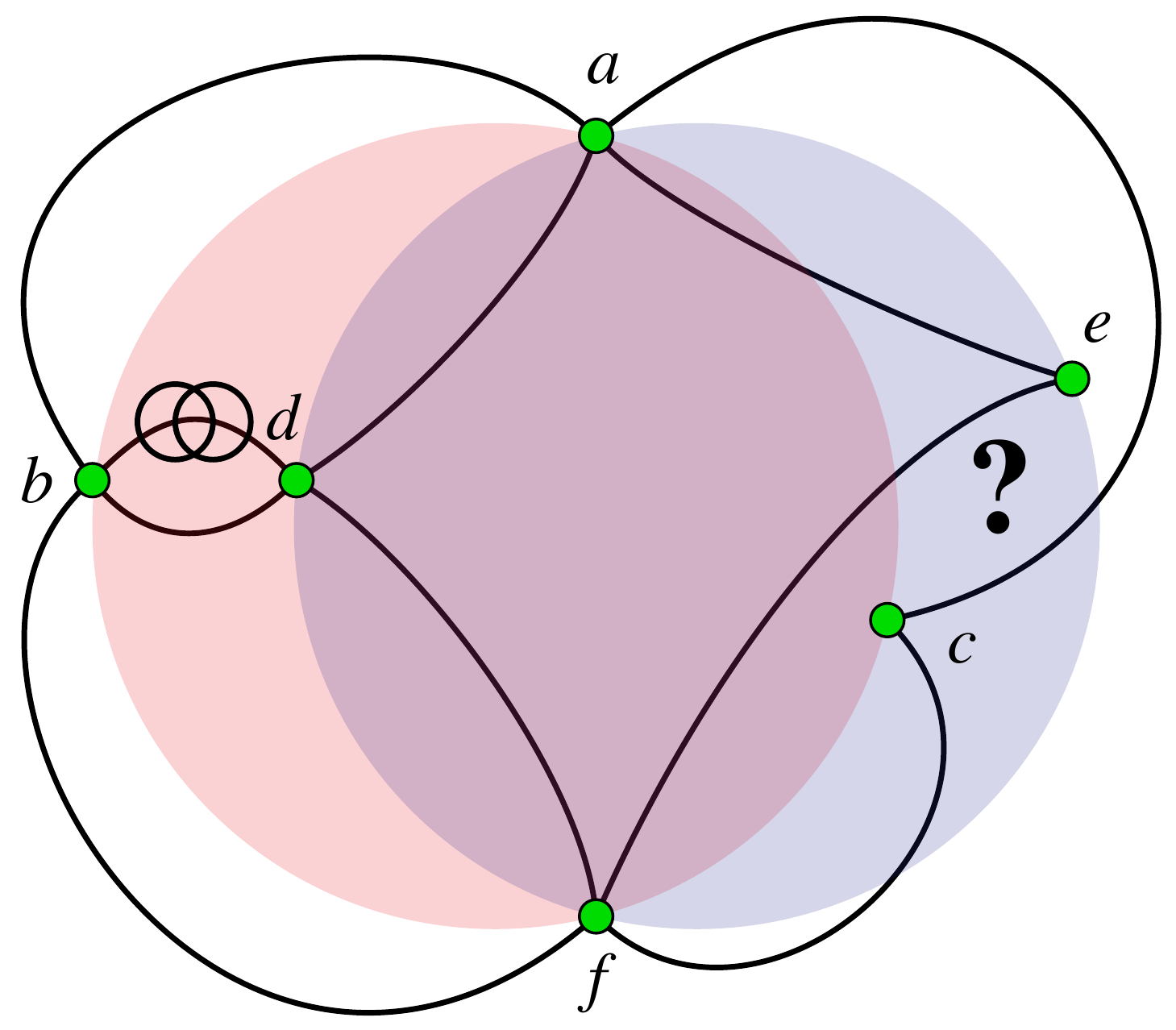}
\caption{Figure for proof of Theorem~\ref{thm:nonlom}}
\label{fig:nonlom-impossibility}
\end{figure}

\end{document}